\documentclass[
 amsmath,amssymb,
 aps, physrev,
]{revtex4-2}

\usepackage{graphicx}
\usepackage{dcolumn}
\usepackage{bm}
\usepackage{soul}

\usepackage{amsthm}
\usepackage{xcolor}
\usepackage{physics}
\setstcolor{red}
\newtheorem{proposition}{Proposition}
\newtheorem{theorem}{Theorem}

\begin{document}

\preprint{APS/123-QED}

\title{\textbf{Demonstration of sequential processors with quantum advantage and analysis of classical performance limits} 
}%

\author{Shota Tateishi$^{1}$}

\author{Wenhao Wang$^{2}$}

\author{Baptiste Chevalier$^{2}$}

\author{Takafumi Ono$^{1}$}
 \email{ono.takafumi@kagawa.ac.jp} 
 
\author{Masahiro Takeoka$^{2,3}$}
\email{takeoka@elec.keio.ac.jp}

\author{Wojciech Roga$^{2}$}
 \email{wojciech.roga@keio.jp}

\affiliation{%
$^{1}$Program in Advanced Materials Science
Faculty of Engineering and Design,
Kagawa University,
2217-20 Hayashi-cho, Takamatsu, Kagawa
761-0396, Japan} 

\affiliation{%
$^{2}$Department of Electronics and Electrical Engineering, Keio University, 3-14-1 Hiyoshi, Kohoku-ku, Yokohama 223-8522, Japan} 

\affiliation{%
$^{3}$Advanced ICT Research Institute, National Institute of Information and Communications Technology (NICT), Koganei, Tokyo 184-8795, Japan
}








\date{\today}

\begin{abstract}
In this paper, we theoretically and experimentally analyze sequential processors with limited communication between parts. We compare the expressivity of sequential quantum and classical processors under the same constraints. They consist of three or four modules, each of which processes local data. The modules of the quantum processor are linked through one-qubit or one-qutrit communication, while those of the classical processor communicate through one bit or one trit.
For the classical processor, we prove bounds on its performance in terms of inequalities on correlations of the output with a target function. We theoretically show that the quantum processor violates these inequalities. We show this violation experimentally on a silicon photonics setup. We describe how to find the classical bound on correlations with arbitrary target function by reducing the problem to the minimization of an Ising-type spin-glass Hamiltonian. Our theory is applicable in general problems, such as the low-rank binary matrix approximation. 
\end{abstract}

\maketitle

\section{Introduction}\label{sec:background}


Fast progress in quantum technology in general and quantum computing in particular urges us to identify the features responsible for the fundamental difference between classical and quantum processors. In addition, quick progress in applications such as machine learning has raised questions about the benefits of quantum processing units. Therefore, a rigorous comparison and meaningful benchmarking of classical and quantum processors, even with simple structure, is crucial for further development of the field. 

In Ref. \cite{Peres}, the authors studied a simple one-qubit processor as an effective classifier. The classical data points $X_i$ were encoded in classically tuned parameters of the circuit gates (see Fig. \ref{fig:resultsa1}). The free parameters of the circuit allowed for training the processor so that the probability measured in the final for a chosen state could be used to distinguish the correct classification label. 
\begin{figure}[h]
    \centering
    \includegraphics[scale=0.6]{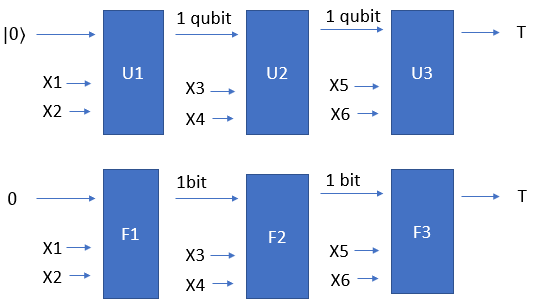}
    \caption{Single-qubit and single-bit sequential processors of similar structure.}
    \label{fig:resultsa1}
\end{figure}
The authors proved the universality of such a classifier, that is, an analog of the universal approximation theory in a single-layer neural network. This showed that even a single-qubit system can compute complicated functions of circuit parameters, assuming that this system is long enough and is appropriately trained. These ideas were experimentally demonstrated \cite{Dutta, Ono, Abe2025,Tolstobrov2024}.

Universality holds for both a quantum and a classical processor if they are sufficiently extended. 
In such a situation, it is natural to ask {\it if there is an advantage of a quantum single-qubit processor over a fairly constructed classical counterpart with fixed finite size}. 
The authors in Ref. \cite{Peres} discussed the possible advantage from entanglement and in Ref. \cite{Roga} from parallelization. In the context of communication complexity \cite{Trojek2005, Brukner2004}, the authors formulated a similar problem 
showing a quantum advantage in sending information in qubit states over the analogous scheme with classical bits. Similar conclusions appeared in the context of temporal entanglement and temporal Bell inequalities \cite{Leggett, Brukner2009}. In another paper \cite{Zukowski2012}, the author formulated the problem in the context of macro-realism and showed a quantum advantage in speeding up the computation of certain functions. In Ref. \cite{Trojek2005}, an experiment was performed on bulk optics showing that a quantum multiparty communication line with a single qubit of communication  can output a binary target function of local parameters $\{X_i\}$. The same target function is not achievable in a single-bit communication classical setup. The setup of this experiment was similar to that in Fig. \ref{fig:resultsa1}.
In addition, in Refs. \cite{Bravyi,Maslov2021} the authors showed  
that certain deterministic outputs from quantum processors cannot be generated by classical circuits in limited space. 
In this sense, one can speak of a larger expressive power of the quantum processor as in Refs. \cite{Gao2022, Anschuetz2023}.

To show these results, it was necessary to prove rigorous limits on how well some functions could be approximated by the output of classical fixed-structure processors. For the functions shown, the classical limit was known or established by the authors. However, for arbitrary functions this classical limit is generally unknown. Thus, a complete understanding of the advantage of the quantum processor over the classical processor even in limited space is still a subject of research. 

Here, we study some aspects of this problem posing the following questions: 
\begin{enumerate}
\item Is there an advantage in using a fixed-length quantum sequential processor with a single-qubit (or single-qutrit) communication between processing modules, with respect to the classical processor of the same structure processing a single bit (or trit)? 
\item Can we quantify the advantage rigorously?
\item Can we demonstrate the advantage experimentally with sufficient accuracy  on nowadays available platforms? In particular, in this work we focus on a silicon photonic platform.  
\end{enumerate}



Our original results are summarized as follows:

\begin{enumerate}
\item We report an experimental realization of the qubit and qutrit sequential processors on the silicon photonic chip. We demonstrate quantum processors that achieve an advantage in computing certain functions of the input variables. The demonstration of a qubit processor is analogous to the previous bulk optics experiment \cite{Trojek2005}. However, it involves different technology, the use of silicon photonics, and different encoding---a single photon in two modes. 
\item To face specific conditions of the design of the silicon photonics chip with four stages and qutrit communication (two photons in two modes), we needed to develop the following method. We sought for the smallest Hamming distance between the deterministic part of the ideal quantum processor output and the optimal output of the corresponding classical processor. We established this bound based on the properties of the quantum processor output and optimization in the space of binary functions.
\item We showed that the problem of minimal distance between an arbitrary binary target function and the output of a classical sequential processor with limited communication can be reduced to optimization in the Ising model. 
\item We successfully solved the Ising problems from the previous point for the classical processors with given target functions. Our models involved up to 48 spins with multipartite interactions. The models were solved on a Fixstars Amplify simulated annealing machine.
\item We indicate consequences of mapping the problem of the smallest distance between the output of a classical sequential processor and a given target for such problems as binary low-rank matrix completion and low-rank binary tensor approximation. 
\end{enumerate}

\section{Sequential processor with limited communication}\label{sec:limited}

Consider a sequential processor consisting of several data processing units (see Fig. \ref{fig:resultsa1}). In each stage, the unit processes the information obtained from the earlier section together with the input provided only locally at this stage, $X_i$, and unknown to other parts. We assume that the information transferred between the modules is strongly restricted. The last module should output a value that depends on all bits of data $\{X_i\}$ inputted at each stage, $O(\{X_i\})$. The design of a sequential logic processor consists of optimizing each unit so that the output function approaches the desired target function $T(\{X_i\})$. The latter can be interpreted as a truth table for a logical function with Boolean variables $\{X_i\}$.

In Fig. \ref{fig:resultsa1}$(a)$ we show a three-module quantum sequential processor with two classical bits $(X_i,X_{i+1})$ processed in each stage and a single-qubit inter module communication. Figure \ref{fig:resultsa1}$(b)$ shows its classical counterpart with, this time, a one-bit communication. In the modules of the quantum processor, we allow for unitary transformations, measurements, and classical functions. Therefore, the quantum processor from Fig. \ref{fig:resultsa1}$(a)$ is at least as powerful in any task as the corresponding classical processor. In the classical circuit, in addition to the restriction that communication between units is limited, the local functions performed by the modules do not have any additional restrictions. 

The goal of our research is to show an advantage of quantum processor with respect to the classical one with similar constraints. Consider the following toy example as an illustration. Assume that we want to design a simple distributed circuit that calculates the sum of several two-bit numbers modulo 4. We also assume that the total sum is promised to be either 0 or 2. We consider that the input numbers are loaded to the circuit in different modules and each of them can only communicate by sending a single bit, or a single qubit, to the next unit. The quantum processor with a single-qubit communication can find the exact solution to the above problem. The scheme is shown in Fig. \ref{fig:addmod4}. The strategy is to choose a subset of four points on a large circle on the Bloch sphere. We encode the two-bit numbers $x$ as a $x\pi/2$ rotation acting on the state transmitted by the previous module without measurement. In this way, the quantum processor can successively rotate the state vector between the four points performing addition of two-bit numbers modulo four. The promise that the total sum is in a given one-bit subset (0 or 2) implies the deterministic output of the correct result if the final state is measured in the corresponding basis. This procedure can be applied for any set of numbers satisfying the promise. The correlation computed between the target and the output of the quantum processor is maximal.

For the classical processor with one-bit communication, at each state, one must decide the state of the single bit to be passed to the next module. As a two-bits number is needed to specify the partial sum modulo 4, the classical processor that passes only one bit of information results in a possible error introduced at each stage. Therefore, this processor cannot output the correct solution of the problem for different sets of input numbers with certainty. This causes reduction of the correlation between the output and the target when the input is chosen randomly.

\begin{figure}[h]
    \centering
    \includegraphics[scale=0.4]{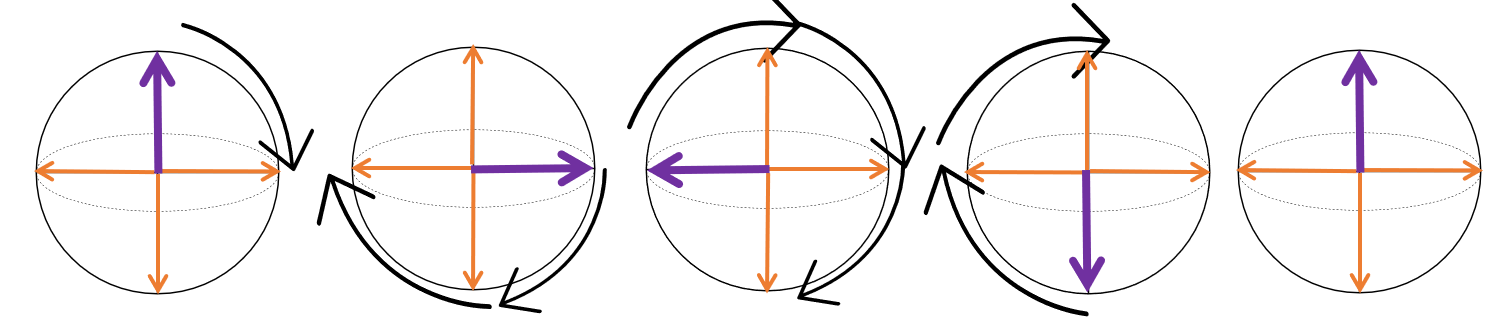}
    \caption{Example. Adding four two-bit numbers $(0,2,3,2)$ modulo 4 with a single-qubit quantum processor. The processor uses operations that rotate a qubit state vector, shown here as a thick violet arrow, between four points, indicated by thin yellow arrows, on the Bloch sphere. The added numbers can be arbitrary two-bit numbers with only one constraint that the total sum modulo 4 is either 0 or 2. For any set of two-bit numbers with the above promise, the ideal quantum processor solves the addition problem without errors.}
    \label{fig:addmod4}
\end{figure}

It is easy to imagine other functions that could be solved by simple quantum processors. For example, adding three-bit numbers modulo 8 with appropriate promise. Although the class of target functions and problems that can be solved by simple or advanced quantum processors is not fully known, we can say that this class is non trivial and contains interpretable examples interesting from a perspective of logic, arithmetic, communication complexity, and fundamental aspects of quantum mechanics \cite{Trojek2005, Brukner2004, Bravyi, Maslov2021}. 

In this paper, we consider deterministic operations with binary targets. Allowing targets to be achievable probabilistically opens yet another space for comparison between classical and quantum processors with fixed constraints. In that case, separation between performances of classical and quantum processors would result in different probability distributions as functions of the input variables. The distributions that would be unachievable by classical processors would lead to violation of variations of Bell-type inequalities established for given constraints. Some classes of such inequalities are well studied. So, the phenomenon of quantum advantage in achieving certain correlations in restricted systems is typical in the sense of not being limited to small number of separated examples. However, studying typicality, which would require different approaches, is not a subject of the present research. We focus on a few target functions tailored to our experimental setup for simplicity. This is enough to prove the existence of an advantage and technical feasibility of its demonstration. It should not be considered as a fundamental restriction. Other setups can be designed to demonstrate similar advantage and general procedures we worked out can be helpful for establishing rigorous results in other cases.


\subsection{Quantum one-qubit processor}

We consider a specific three-stage quantum processor with single-qubit communication between modules and two classical bit inputs $(X_{2i},X_{2i+1})$ at each module. Depending on the two input bits, we choose one of the four one-qubit gates to be applied to the qubit obtained from the earlier part:
\begin{eqnarray}
(0,0) &\rightarrow& \hat{1}=\begin{bmatrix}1&0\\0&1\end{bmatrix},\label{gate1}\\
(0,1) &\rightarrow& \hat{\sigma}_x=\begin{bmatrix}0&1\\1&0\end{bmatrix},\\
(1,0) &\rightarrow& \hat{H} =\frac{1}{\sqrt{2}}\begin{bmatrix}1&i\\i&1\end{bmatrix},\\
(1,1) &\rightarrow& \hat{H}^*=\frac{1}{\sqrt{2}}\begin{bmatrix}1&-i\\-i&1\end{bmatrix}.\label{gate4}
\end{eqnarray}

The initial qubit state is $|0\rangle$. In the end, the qubit is measured in the computational basis $\{|0\rangle,|1\rangle\}$. For all $2^6$ input bit values $X_1,...,X_6$ ordered in lexicographical order, 
the final measurement of our ideal processor returns either deterministic outputs $|0\rangle$ or $|1\rangle$, or a state which is a superposition of the two basis states. We define the target function as follows:
\begin{align}
T(\{X_i\})&=1 \text{ if the output is deterministically } |0\rangle,\\
T(\{X_i\})&=-1 \text{ if the output is deterministically } |1\rangle,\\
T(\{X_i\})&=0 \text{ if the output state is a superposition of } |0\rangle \text{ and } |1\rangle.
\end{align}
The target function, after being reshaped into a matrix, for convenience, can be written explicitly as follows:
\begin{widetext}
\begin{equation}
T(X_1X_2,X_3...X_6)=\left[\begin{array}{cccccccccccccccccccc}
-1&1&0&0&1&-1&0&0&0&0&1&-1&0&0&-1&1\\
1&-1&0&0&-1&1&0&0&0&0&-1&1&0&0&1&-1\\
0&0&1&-1&0&0&-1&1&1&-1&0&0&-1&1&0&0\\
0&0&-1&1&0&0&1&-1&-1&1&0&0&1&-1&0&0
\end{array}\right],\label{targetbit}
\end{equation}
\end{widetext}
where the first two bits index the row and the remaining four bits index the column. This notation is similar to the Karnaugh maps \cite{karnaugh} known in the literature on logic design. Alternatively, it can be seen as a tensor, and this is how we refer to it in Sec. \ref{sec:lowranks}.

We used this target function to calculate the correlation with the experimental output with real experimental imperfections. The experiment and results are presented in Sec. \ref{sec:silicon}. 
There, based on the experimental output $O(x_i)$ of each of $N$ randomly selected sequences $x_i$ of six bits $X_1,...,X_6$, we calculate the correlation function
\begin{equation}
C=\frac{1}{N}\sum_{i=1}^N O(x_i)T(x_i).
\end{equation}

In Appendix A, 
we prove that the bound on the achievable correlation function for an optimized classical sequential processor with three stages, and one-bit communication is 0.25. This corresponds to the minimum number of errors in the outputs of the optimal classical processor with respect to the non zero part of the target equal to 8. An optimal strategy for this processor is given in Table 1. 

\begin{table}[]
\begin{tabular}{|c|c|c|}
\hline
$i$ & $F_{i-1}=0$ & $F_{i-1}=1$ \\ \hline
$F_1$ & (0,1,1,1) & -  \\ \hline
$F_2$ & (0,1,1,0) & (1,0,0,1)  \\ \hline
$F_3$ & (-1,1,-1,1) & (1,-1,1,-1)  \\ \hline
\end{tabular}\caption{Optimal strategies for one-bit processor with three nodes. The symbols correspond to the notation from Fig. \ref{fig:resultsa1}. For example, here $(0,1,1,0)$ are bits communicated by the second node when the local inputs are $(00,01,10,11)$ respectively, and the bit $0$ is obtained from the previous node.}
\label{table:bit3} \end{table}

The processor described above realizes the adder of two-bit numbers modulo 4 from the previous example. Here, bits $00$ are associated with number 0, and $10$, $01$, and $11$ with 1, 2, and 3, respectively. The zeros of the target function correspond to the situation when the promise from the problem formulation is not satisfied. Target functions with many zeros, as in Eq. (\ref{targetbit}) are typically considered in logic circuit design, where the truth table is specified only on desired entries. Leaving the remaining entries not specified allows for more freedom in designing the optimal circuit. Therefore, not fully specified target functions should not be considered as a restriction. 

\subsection{Quantum one-qutrit processor}

We perform a similar analysis for a quantum sequential processor with three and four modules and a communication carried forward by a three-level quantum system (qutrit). Our analysis is tailored for a quantum optics experiment on silicon photonics, so the logical state of the qutrit we consider is physically encoded by the state of two photons in two different modes.
Each logical state is uniquely described by a physical state in the two-mode Fock basis as follows:
\begin{align*}
    \ket{0}_L &\rightarrow \ket{20},\\
    \ket{1}_L &\rightarrow \ket{02},\\
    \ket{2}_L &\rightarrow \ket{11},
\end{align*}
where $|20\rangle$ states for two photons in the first mode and zero photons in the second, etc.  

As the input state, we take $|20\rangle$. As gates, depending on each module input bits $(X_{2i},X_{2i+1})$, we choose the following:

\begin{itemize}
\item For the three-stage processor and for modules 2, 3, and 4 of the four-stage processor, we applied gates [Eqs. (\ref{gate1})-(\ref{gate4})] acting on the photon creation operators in each mode.
\item For module 1 of the four-step processor, the design of the circuit of our chip did not allow for applying the same gates as in the other modules. Instead, we applied the following gates
\begin{eqnarray}
\label{eq:qutrit_unitary_1}
(0,0) &\rightarrow& \frac{1}{\sqrt{2}}\begin{bmatrix}1&i\\-i&-1\end{bmatrix},\\
\label{eq:qutrit_unitary_2}
(0,1) &\rightarrow& \begin{bmatrix}0&1\\-1&0\end{bmatrix},\\
(1,0) &\rightarrow&\begin{bmatrix}-1&0\\0&1\end{bmatrix},\\
\label{eq:qutrit_unitary_4}
(1,1) &\rightarrow& \frac{1}{\sqrt{2}}\begin{bmatrix}1&-i\\i&-1\end{bmatrix}.
\end{eqnarray}

\end{itemize}
If the final measurement of an ideal processor deterministically indicates the state $|20\rangle$, we define the entry of the corresponding target function as 1. If it deterministically indicates the state $|02\rangle$, we define the entry $-1$. Otherwise, the target function entry is set to $0$. The quantum processor we built does not allow us to achieve state $|11\rangle$ deterministically from the initial state $|20\rangle$ by any unitary transformation that preserves the bosonic commutation relations. Therefore, the target function of our qutrit processor resembles a qubit processor target. 

Notice that this physical restriction stands as a limitation of the quantum processor power. Indeed, the access to the logical state $\ket{2}_L$ is forbidden by the law of the physical system used for implementation. Thus, despite this constraint, showing an advantage of a single-qutrit processor with respect to its single-trit counterpart is an even stronger evidence in favor of the quantum advantage.

The target of the four part processor after reshuffling as a matrix, for convenience, is explicitly given as
\begin{widetext}
\begin{equation}
T(X_1...X_4,X_5...X_8)=\left[\begin{array}{cccccccccccccccccccc}
0&0&-1&1&0&0&1&-1&-1&1&0&0&1&-1&0&0\\
0&0&1&-1&0&0&-1&1&1&-1&0&0&-1&1&0&0\\
-1&1&0&0&1&-1&0&0&0&0&1&-1&0&0&-1&1\\
1&-1&0&0&-1&1&0&0&0&0&-1&1&0&0&1&-1\\
1&-1&0&0&-1&1&0&0&0&0&-1&1&0&0&1&-1\\
-1&1&0&0&1&-1&0&0&0&0&1&-1&0&0&-1&1\\
0&0&-1&1&0&0&1&-1&-1&1&0&0&1&-1&0&0\\
0&0&1&-1&0&0&-1&1&1&-1&0&0&-1&1&0&0\\
0&0&1&-1&0&0&-1&1&1&-1&0&0&-1&1&0&0\\
0&0&-1&1&0&0&1&-1&-1&1&0&0&1&-1&0&0\\
1&-1&0&0&-1&1&0&0&0&0&-1&1&0&0&1&-1\\
-1&1&0&0&1&-1&0&0&0&0&1&-1&0&0&-1&1\\
-1&1&0&0&1&-1&0&0&0&0&1&-1&0&0&-1&1\\
1&-1&0&0&-1&1&0&0&0&0&-1&1&0&0&1&-1\\
0&0&1&-1&0&0&-1&1&1&-1&0&0&-1&1&0&0\\
0&0&-1&1&0&0&1&-1&-1&1&0&0&1&-1&0&0\\
\end{array}\right].
\label{target4}
\end{equation}
\end{widetext}

The three-stage processor has the target function the same as the first $4\times 16$ block of this matrix.

\begin{table}[]
\begin{tabular}{|c|c|c|c|}
\hline
$i$ & $F_{i-1}=0$ & $F_{i-1}=1$ & $F_{i-1}=2$  \\ \hline
$F_1$ & (0,1,2,1) & - & - \\ \hline
$F_2$ & (2,1,2,1) & (1,2,0,2) & (2,1,1,2) \\ \hline
$F_3$ & (-1,1,-1,1) & (1,-1,-1,1) & (-1,1,1,-1) \\ \hline
\end{tabular}\caption{Optimal strategies for one trit processor with three nodes.}
\label{Tabstage3} \end{table}

\begin{table}[]
\begin{tabular}{|c|c|c|c|}
\hline
$i$ & $F_{i-1}=0$ & $F_{i-1}=1$ & $F_{i-1}=2$  \\ \hline
$F_1$ & (0,2,1,0) & - & - \\ \hline
$F_2$ & (1,2,0,1) & (2,1,2,0) & (2,0,1,2) \\ \hline
$F_3$ & (1,0,1,2) & (2,1,1,0) & (1,2,0,1) \\ \hline
$F_4$ & (1,-1,-1,1) & (-1,1,1,-1) & (-1,1,-1,1) \\ \hline
\end{tabular}\caption{Optimal strategies for one trit processor with four nodes.}
\label{Tabstage4} \end{table}

As proven in Appendix \ref{app:classical}, for the three-stage trit processor the classical bound on the correlation with the target function is 0.4375. This corresponds to at least two errors between the output of the optimized processor and the non zero output of the target function. The explicit strategy of such processor is shown in Table \ref{Tabstage3}. For the four-stage processor, the classical limit on the correlation with the above target is 0.375, which corresponds to at least 16 errors in the optimized output. This strategy is shown in Table \ref{Tabstage4}.




\section{Silicon photonics experiment}\label{sec:silicon}

We implemented both a one-qubit and a one-qutrit processor on a silicon photonic integrated circuit \cite{Wang2019,Moody2022,Alexander2025} and evaluated their performance using the correlation function described in Sec. \ref{sec:limited}. As shown in Fig. \ref{fig:setup}(a), the experimental setup employed a $4\times4$ universal unitary transformation circuit composed of beam splitters and phase shifters \cite{Clements2016,Bogaerts2020}. In this study, we used two spatial modes of the circuit to construct a sequential processor by cascading multiple unitary transformations. Each unitary transformation consisted of a phase shifter and a Mach–Zehnder interferometer, allowing us to realize the target operations with high precision.

\begin{figure}[h]
    \centering
    \includegraphics
    [scale=0.3]
    {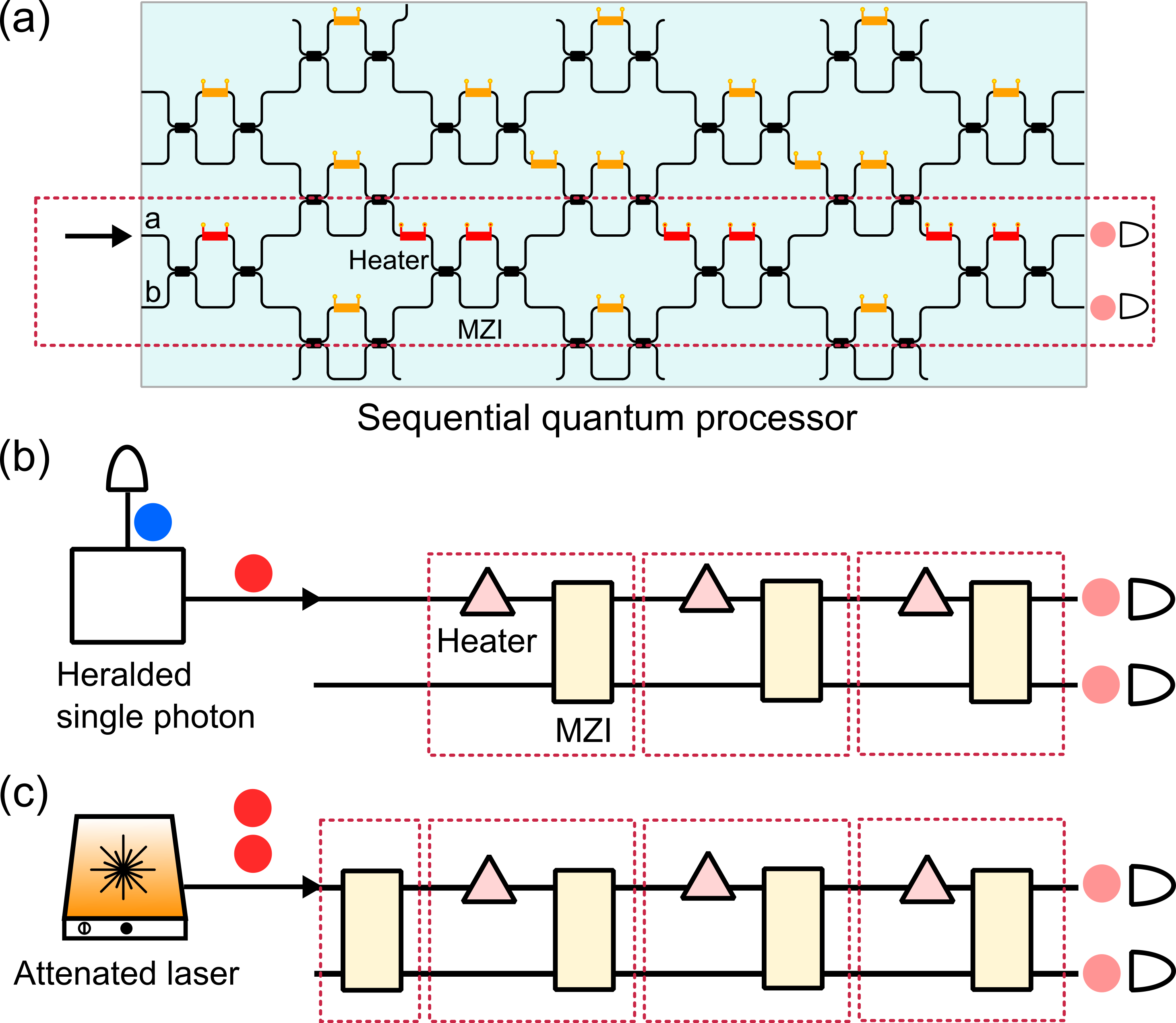}
    \caption{(a) Schematic of the sequential quantum processor implemented on a silicon photonic chip. The unitary gate was implemented using a heater and an Mach–Zehnder interferometer enclosed by the red dotted lines. (b) Overview of the one-qubit processor. A heralded single photon was injected into the sequential processor. (c) Overview of the one-qutrit processor. Two photons generated from attenuated laser light at the single-photon level were injected into the sequential processor. In panels (b) and (c), the circuit enclosed by the red dotted lines corresponds to a single unitary gate in each case.}
    \label{fig:setup}
\end{figure}

For the implementation of the one-qubit processor, heralded single photons were employed, as illustrated in Fig. \ref{fig:setup}(b). Photon pairs were generated via spontaneous four-wave mixing (SFWM), utilizing the third-order nonlinearity of silicon \cite{Sharping2006,Clemmen2009}. Specifically, a continuous-wave (CW) pump light with a wavelength of 1549.32 nm and an output power of approximately 10 mW was injected into a silicon waveguide. This generated signal and idler photons through the SFWM. The signal photons had a central wavelength of 1559.0~$\pm$~0.8 nm, while the idler photons had a central wavelength of 1539.8~$\pm$~0.8 nm. The spectral width of these photons is determined by the bandwidth of the off-chip frequency filter (200 GHz). The signal photons were used as heralds (triggers), and the idler photons were directed to the processor circuit \cite{Ma2017,Faruque2018,Faruque2019}. After passing through these filters, both the idler and the corresponding signal photons were detected using superconducting nanowire single-photon detectors. The generation rate of signal single photons obtained after passing through the circuit was approximately 30 counts per second.

In this experiment, the input state was $|10\rangle_{ab}$, where a single photon was present in mode $a$. Randomly selected operations from the set $\{ \hat{I}, \hat{\sigma}_z, \hat{H}, \hat{H}^* \}$ were applied through three successive unitary transformations. At the output, the occurrence probabilities of the states $|10\rangle_{ab}$ and $|01\rangle_{ab}$ were measured from the photon number distribution obtained over a 1-s integration period. Using approximately 85000 detection events, the correlation function with respect to the target state was calculated.

Figure~3 shows a histogram of the obtained correlation function. To construct the histogram, the approximately 85000 experimental data points were divided into subsets of 1000 events each. The correlation function was calculated for each subset, and the frequency distribution of these values was plotted. From the histogram, the mean of the correlation function was 0.490, and the square root of the variance (standard error of the mean) was 0.017. These results indicate that the experimentally obtained correlation function significantly exceeds the limit of 0.25 achievable by the classical processor with one-bit communication. In the figure, the vertical blue line corresponds to this limit. 

Next, for the implementation of the one-qutrit processor, we attenuated a laser beam to the single-photon level and approximated a two-photon state by post-selecting events in which two photons were detected simultaneously [Fig.~\ref{fig:setup}(c)]. 
Specifically, each of the two output ports of the photonic integrated circuit was connected to an off-chip fiber beam splitter, followed by single-photon detectors at the two outputs of each splitter. Coincidence detections of one photon at each detector were used to identify events corresponding to the $\lvert 20 \rangle_{ab}$ or $\lvert 0;2 \rangle_{ab}$ states. This detection configuration functions as a pseudo-photon-number-resolving scheme and allows us to detect photon bunching in each mode.

To suppress multi photon contributions, the input laser was strongly attenuated such that the single-count rate was approximately $2.5\times10^{5}$ counts per $0.5~\mathrm{s}$. Under this condition, we observed approximately $5000$ coincidence events after the beam splitter (corresponding to about $10000$ events before the splitter). The mean photon number of the attenuated coherent state was estimated to be $\sim 0.08$. Given that the photon-number statistics follow a Poisson distribution, the contribution of three-photon or higher-order events to the measurement results is estimated to be $\sim 0.008\%$. Therefore, the impact of multi photon components introduces a negligibly small error of below $0.01\%$ in the experimental data.

The input state for the qutrit experiment was $|20\rangle_{ab}$, corresponding to two photons in mode $a$. Four unitary transformations, randomly selected from the operations defined in Eqs.~(\ref{gate1})–(\ref{gate4}) and (\ref{eq:qutrit_unitary_1})–(\ref{eq:qutrit_unitary_4}), were sequentially applied. The output probabilities of the states $|20\rangle_{ab}$ and $|02\rangle_{ab}$ were measured, and approximately 65,000 detection events were collected. Correlation functions were then calculated based on these data.

Figure~4 shows a histogram of the correlation function for a four-module quantum processor with one qutrit communication. As in Fig.~3, approximately 65000 experimental data points were divided into subsets of 1000 events each, and the correlation function was calculated for each subset. The histogram displays the frequency distribution of the resulting correlation values with respect to the target function (\ref{target4}). From the histogram, the mean of the correlation function was 0.525, and the square root of the variance was 0.014. These results indicate that the experimentally obtained correlation function significantly exceeds the classical limit which, as shown in Fig. \ref{fig:resultsc}, is 0.375.

\begin{figure}[h]
    \centering
    \includegraphics[scale=0.3]{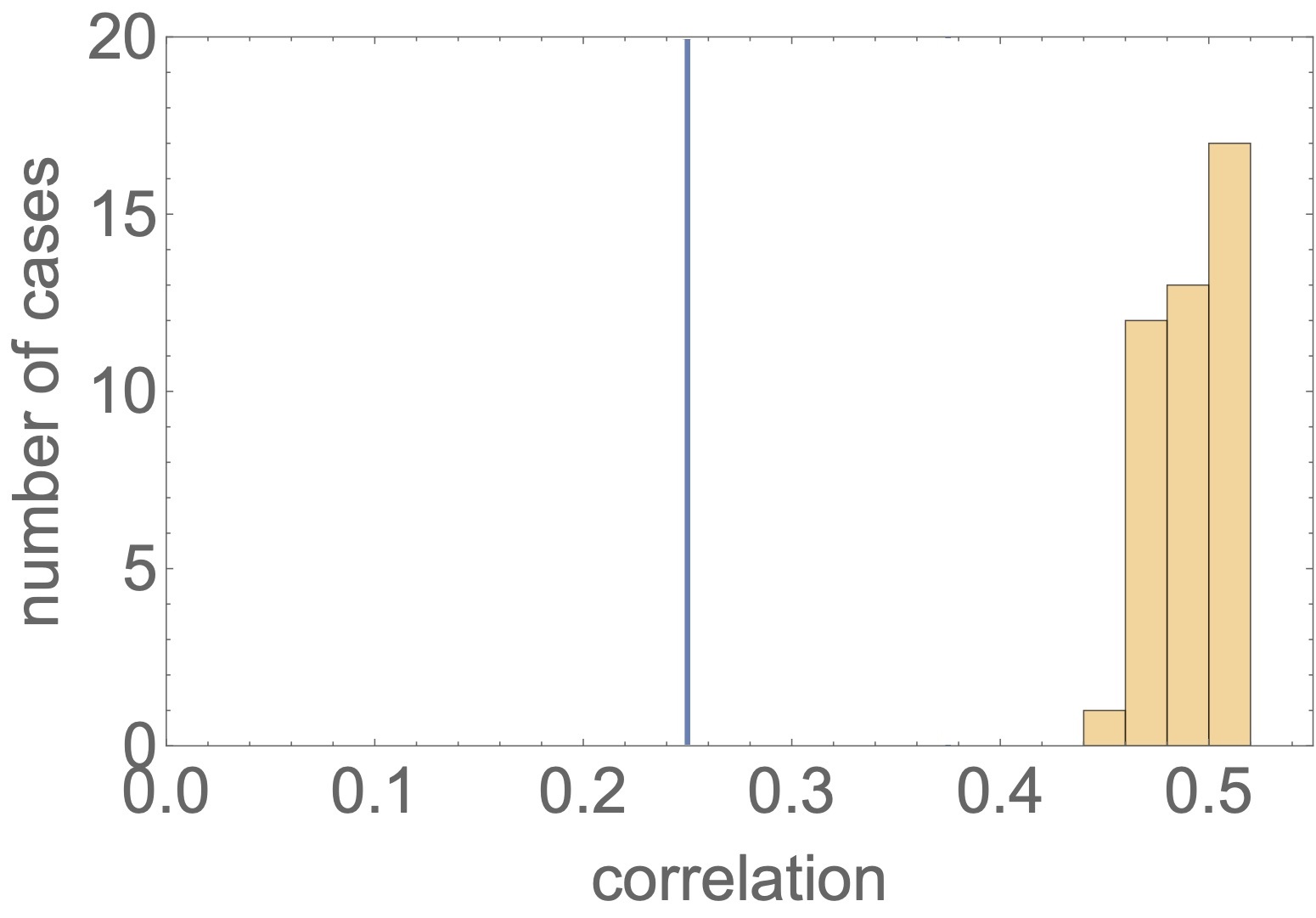}
    \caption{The histogram of the correlation function based on the data from three-module single-qubit processor experiment on silicon photonics. The blue line shows the classical limit for the correlation, which for the considered target function is equal to 0.25. 
    }
    \label{fig:resultsa}
\end{figure}

\begin{figure}[h]
    \centering
    \includegraphics[scale=0.3]{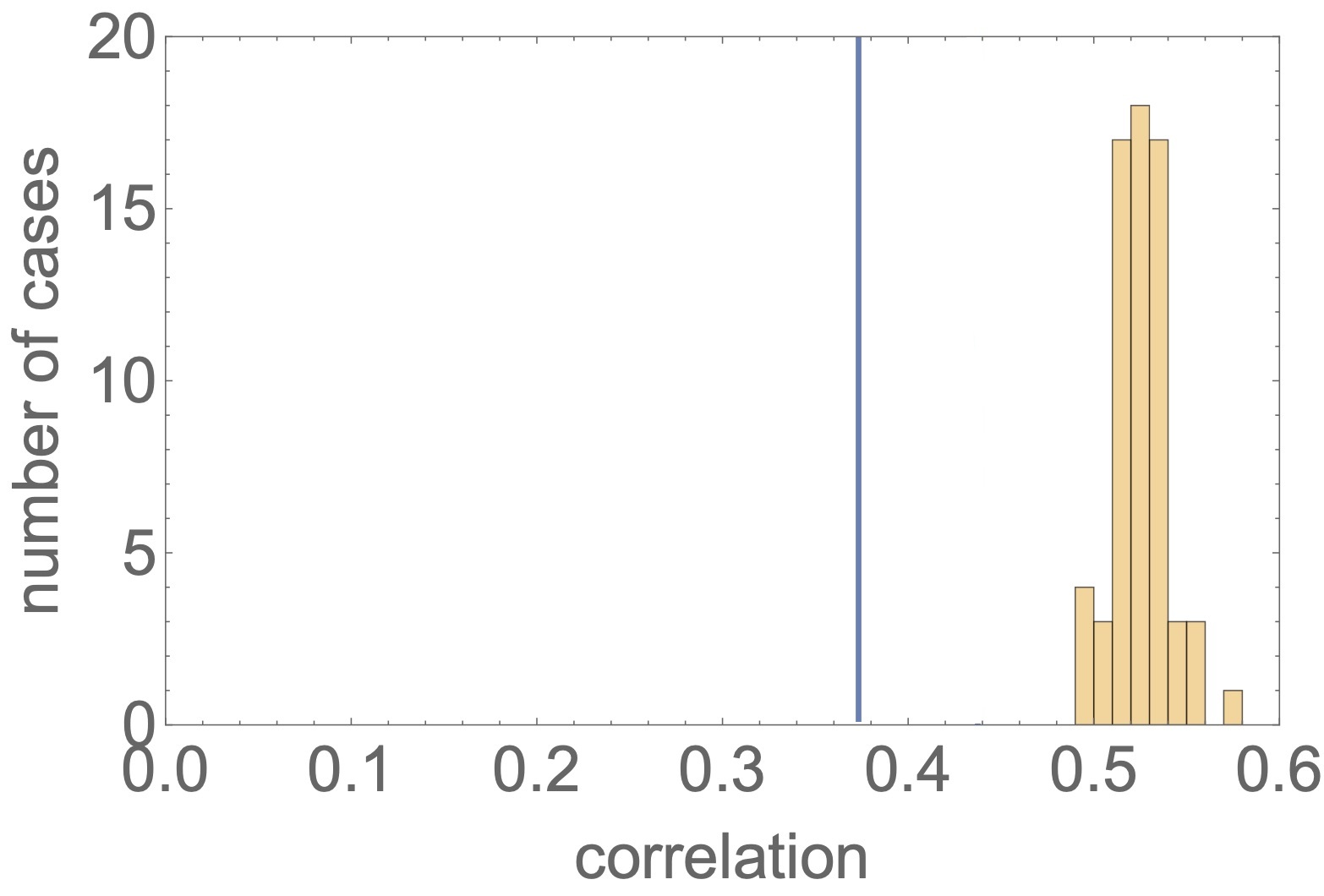}
    \caption{The histogram of the correlation function based on the data from four-module single-qutrit processor experiment on silicon photonics. The blue line shows the classical limit for the correlation, which for the considered target function is equal to 0.375.
    }
    \label{fig:resultsc}
\end{figure}


\begin{figure}[h]
    \centering
    \includegraphics[scale=0.15]{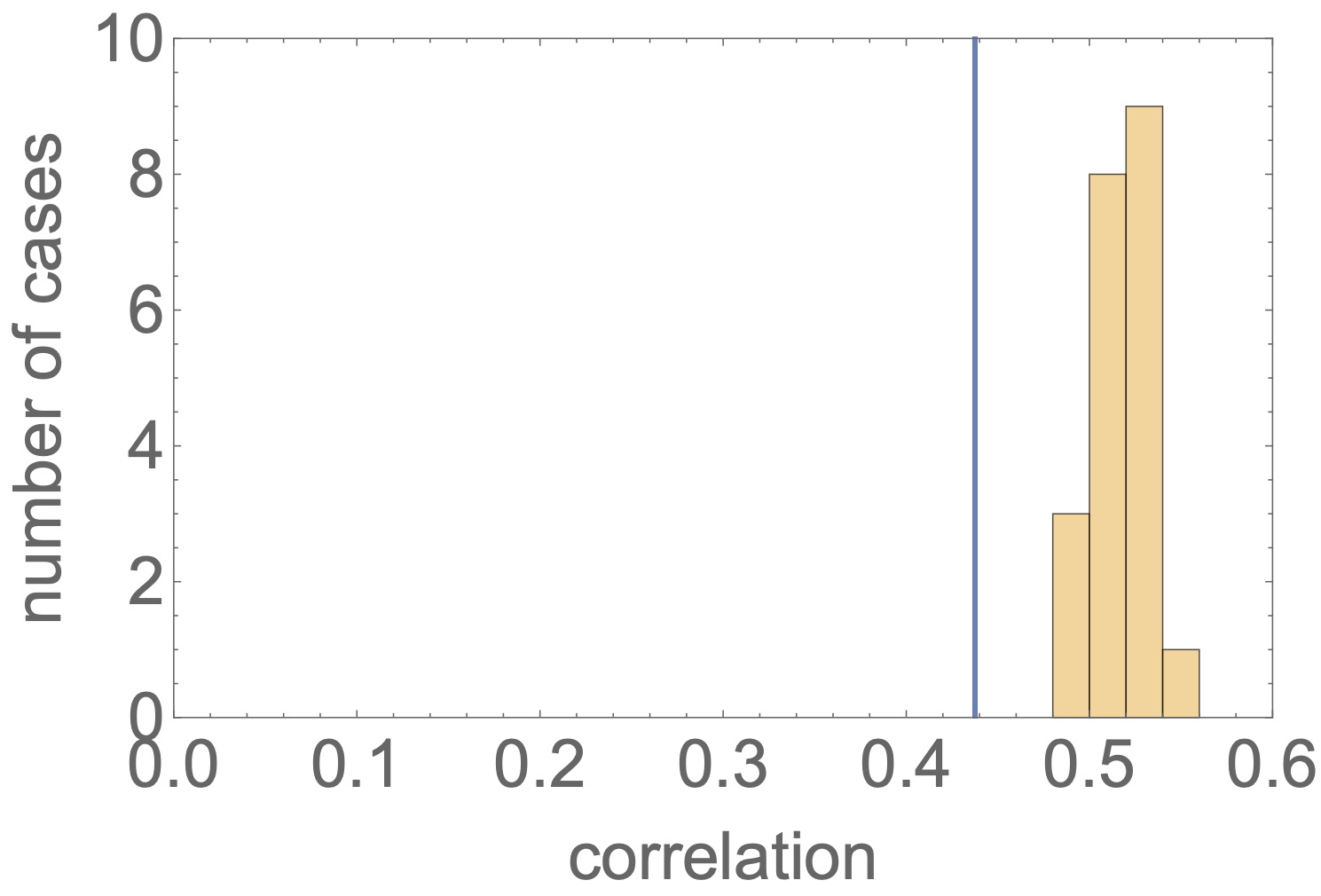}
    \caption{The histogram of the correlation function based on the data from three-module single-qutrit processor experiment on silicon photonics. The blue line shows the classical limit for the correlation which for the considered target function is equal to 0.4375.
    }
    \label{fig:resultsb}
\end{figure}

Finally, in post-processing, from the data of the qutrit experiment with four modules, we extracted data that correspond to the first gate fixed as the identity. In this way, we performed the analysis of the three-module processor with single-qutrit communication. The target function in this case is given by the first four rows of Eq. (\ref{target4}). We built the histogram of the values of the correlation function for 65 subsets of data with approximately 250 points in each subset. The average is 0.509 with the standard error of the average 0.068. This allows us to conclude that the observed value of the correlation significantly exceeds the limit for the correlation achievable by the classical processor with one-trit communication, which, as shown in Fig. \ref{fig:resultsb}, is equal to 0.4375.  

In conclusion, in each case that we analyzed, we observe that the realistic experiment produces the correlation function close to 0.5, which is the maximum value of the correlation function with respect to the target [Eqs. (\ref{targetbit}) or (\ref{target4})] with half of the entries equal to zero. 
The bounds on the achievable correlation function for the optimal classical sequential processors are shown in Figs. \ref{fig:resultsa}--\ref{fig:resultsb} as the vertical blue lines are established in Section \ref{sec:limited}.
We conclude that the classical processor cannot output a function that approaches the target function for any deterministic strategy. In this sense, the quantum processor significantly outperforms the classical one in this task, showing a separation of expressivity between the classical and quantum processors.

\section{Optimal classical sequential processors}\label{sec:optimal}

In the previous sections, we considered the problem of how well a classical sequential processor with single-bit or single-trit communication can approximate a given target function of distributed input data. The target functions considered until now were specific ones generated by the outputs of quantum processors. We showed the gap in expressivity of the classical processor with respect to the quantum processor for specific tasks. The proofs derived in Appendix \ref{app:classical} are based on explicit properties of the target functions.

In this section, we consider a broader problem. Given an arbitrary binary target function, we ask what is the best approximation that can be obtained by the output of a classical sequential processor. The result can be useful to establish rigorous bounds on classical processors in other designs and experimental demonstrations of quantum advantage in limited space processors. 

The problem what is the class of function which can be obtained as an output of a quantum processor depends on specific constraints. We will not answer this question in this paper. However, establishing bounds on classical processor of a given structure is a first step to identify the classes of functions achievable in quantum regime, as we know that all processors that are classical in ideal conditions can be realized by the quantum counterparts.

In general, saying how close a given binary function of Boolean arguments can be approximated by outputs of a classical processor with a fixed structure belongs to a class of high complexity. 
To solve it efficiently, one would need an algorithm akin to solving the low-rank binary matrix approximation problem \cite{Valiant, Lokam, Dan}. However, this problem or its alternatives are known to be NP hard. In this section, we reduce this task to an Ising problem, which can be approachable by nowadays and future Ising solvers based on simulated or quantum annealing, or other techniques and,  under certain conditions, can be efficiently solved.

Consider an $N$-stage sequential processor with one-bit communication between parts and two-bit inputs in each module. The local inputs are not known to other parts. The performance of the processor can be optimized on the strategies of the modules.\\


We formulate two tasks that are addressed in the subsequent sections:
\begin{enumerate}
\item Task 1: Decide if a given binary target function $T(\{X_i\})$ can be achieved as an output from the processor. 
\item Task 2: Determine the smallest Hamming distance between the output of the classical processor $O(\{X_i\})$ 
and the target function $T(\{X_i\})$. 
\end{enumerate}


\subsection{Task 1: Existence criterion}\label{sec:criterion}

We formulate the existence criterion as the following practical procedure proven in Appendix \ref{app:existence}.
Consider a binary function $T$ of several variables. Consider also an arbitrary sequential processor with a single-bit communication and a given distribution of variables $\lbrace X_i\rbrace$ between the modules. The following procedures allow us to check if the sequential processor can output $T$ on $\lbrace X_i\rbrace$:
\begin{enumerate}
    \item Reshuffle the function $T$ into a matrix $M$ whose rows are indexed by the local variables of earlier modules and columns are indexed by the remaining variables.
    \item Compute the row rank of $M$.
    \item Repeat for other reshufflings of $T$.
\end{enumerate}
If, for any $M$, the associated row rank is at most 2, $T$ can be outputted by the sequential processor.


For sequential processors with single-trit communication, the row rank should be 3, and so on.

\subsection{Task 2: Optimal circuit as the Ising problem}

\subsubsection{Example}

\begin{figure}[h]
    \centering
    \includegraphics[scale=0.5]{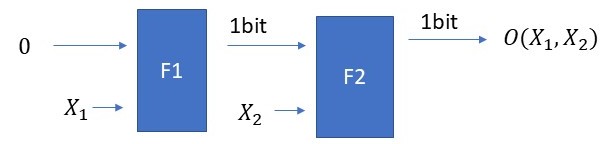}
    \caption{Two-stage sequential processor with one-bit communication and one bit input $X_i$ at each stage. The output function of inputs is denoted by $O(X_1,X_2)$.}
    \label{fig:twoprocessor}
\end{figure}

For simplicity, let us consider an example---a two-module sequential processor with one-bit communication and one-bit input at each stage, as shown in Fig. \ref{fig:twoprocessor}. The first part, depending on the bit $X_1$ it obtains as the input, chooses which binary variable to pass to the next stage. Its strategy is given by a binary function $F_1$ which is defined by binary variables $a_1$ and $a_2$ returned for the respective input $X_1$, i.e., 
$$F_1(X_1) = \begin{cases}
    a_1 \text{ if } X_1=0,\\ a_2 \text{ if } X_1=1.
\end{cases}$$ 
The second module strategy $F_2$ is a function of the bit passed from the previous part and an extra bit $X_2$. This function's output is given by a four-bit binary vector of binary variables $a_3,..., a_6$ (see Fig. \ref{fig:variables}). The goal is to choose the values of the variables $a_1,..., a_6$ such that the output from the last part of the processor, for given variables $X_1$ and $X_2$, was as close as possible to the given binary target function $T(X_1,X_2)$. 

\begin{figure}[h]
    \centering
    \includegraphics[scale=0.4]{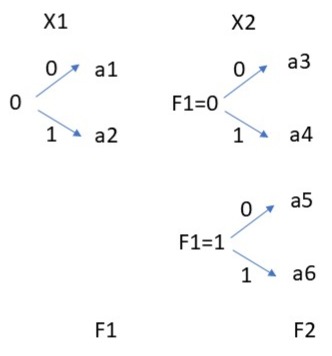}
    \caption{Example of binary variables that define each part's strategy functions $F_i$ of their respective inputs. In this case, there are two parts in the processor. Each part processes the bit obtained from the previous part and one extra bit $X_i$.}
    \label{fig:variables}
\end{figure}

For a two-stage processor, the problem is trivial because the target has only four values. We can associate the first two elements of the target function with $a_3$ and $a_4$, and the other two with $a_5$ and $a_6$ and put $a_1=0$ and $a_2=1$. However, we want to show that we can obtain this result as the optimization of an Ising like problem. Then, we generalize the entire formalism to any number of stages and more general setups in which the problem is highly nontrivial.

For any $X_1$ and $X_2$, the output of the processor $O(X_1,X_2)$ is
\begin{equation}
O(X_1,X_2)=(1-a_{1+X_1})a_{3+X_2}+a_{1+X_1}a_{5+X_2}.
\end{equation}

For binary functions, the minimum Hamming distance coincides with the maximum of the correlation function. In this example, it is given as
\begin{eqnarray*}
C&=&\sum_{\{X_i\}}T(X_1,X_2)O(X_1,X_2)\\
&=&T(0,0)[(1-a_1)a_3+a_1a_5]\\
&+&T(0,1)[(1-a_1)a_4+a_1a_6]\\
&+&T(1,0)[(1-a_2)a_3+a_2a_5]\\
&+&T(1,1)[(1-a_2)a_4+a_2a_6].
\end{eqnarray*}

The last formula is a quadratic form in binary variables. This can be written as an Ising problem with nonlocal two-body interactions. Its maximum provides the set of bits $a_1,...,a_6$ that maximizes $C$.

\subsubsection{Generalization}

The construction of the example above can easily be generalized to arbitrary processor length, arbitrary communication, and arbitrary number of local input bits. Let us give here the explicit formula for a three-stage processor with one-bit communication and a two-bit local input, as in the lower part of Fig. \ref{fig:resultsa}. The binary variables that define the strategies are shown in Fig. \ref{fig:variables3}.

\begin{figure}[h]
    \centering
    \includegraphics[scale=0.45]{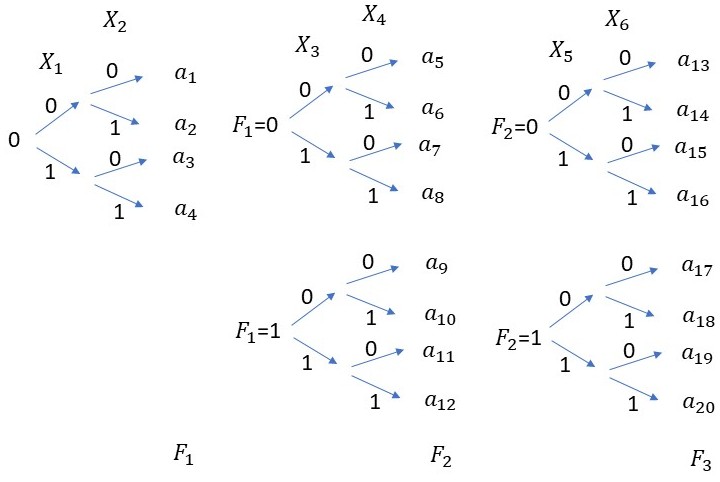}
    \caption{Example of binary variables that define each part's strategy functions $F_i$ of their respective inputs. In this example, there are three modules in the processor. Each part processes a bit obtained from the previous part and two extra bits.}
    \label{fig:variables3}
\end{figure}

The formula for the correlations between the target and the output reads
\begin{eqnarray}
C&=&\sum_{\{X_i\}}T(\{X_i\})\label{correlationbit3}\\
&*&[(1-a_{1+2X_1+X_2})(1-a_{5+2X_3+X_4})a_{13+2X_5+X_6}\nonumber\\
&+&(1-a_{1+2X_1+X_2})a_{5+2X_3+X_4}a_{17+2X_5+X_6}\nonumber\\
&+&a_{1+2X_1+X_2}(1-a_{9+2X_3+X_4})a_{13+2X_5+X_6}\nonumber\\
&+&a_{1+2X_1+X_2}a_{9+2X_3+X_4}a_{17+2X_5+X_6}]\nonumber
\end{eqnarray}

In this case, the correlation function is a polynomial of degree three in the binary variables that can be mapped in a standard way to the Ising model with up to three spins, and possibly far-range interactions.

Further generalizations are needed to consider four-stage processors and the processor that allows for trit-based communication. In this case, variables $a_i$ are integers from the set $\{0,1,2\}$. Because we routinely use solvers acting on binary variables, before formulating the formulas for correlations, we first map each trit variable to two one-bit variables. We do not show the final formulas that can be easily obtained with the same strategies as described above. 

\subsubsection{Solvers}

In the previous sections, we showed that finding the highest correlations between the output of the classical sequential processor with restricted communication and a given binary target function is equivalent to solving an Ising problem for a chain of spins with nonlocal and multi-spin interactions. This problem is NP hard, and efficient solvers, for the general case, do not exist. However, one can use specific heuristics based on approximate Ising solvers.  

In our research, we use a dedicated computer provided by Fixstars Amplify \cite{fixstars}. It finds solutions through the process of simulated annealing. We used the solver to find candidate solutions for the maximum correlation. The correlation is taken between the output of the classical three-module sequential processor with two-bit local variables and one-bit communication and the target given by the deterministic part of the output of the corresponding quantum processor with one-qubit communication. The problem involved 20 binary variables (spins) as shown in Fig. \ref{fig:variables3}. We defined the Ising Hamiltonian corresponding to the correlation function (\ref{correlationbit3}). (The number of terms in the Hamiltonian is roughly the number of nonzero elements in the target function times 2 to the power of the number of modules. Note that this is much less than the number of all possible binary functions which for one-bit processor with two-bit inputs in each module is roughly 2 to the power of eight times the number of modules. The brute force optimization would require this number of steps.) The solver was able to find the solution with eight errors with respect to the target function. as shown in Table \ref{table:bit3}. As proven in Appendix A, this is the lower bound on the number of errors for this target. Therefore, we conclude that the solver found the optimal solution.

We repeated the numerical experiment on Fixstars Amplify solver successfully finding cases corresponding to the optimal strategies of the classical processors with single-trit communication, three and four modules, and the target functions given by the deterministic parts of the corresponding quantum optical processors with qutrit communication as in the quantum optics experiment. The target functions are given in Eq. (\ref{target4}) or the first four rows of it. The lower bounds on the number of errors are derived in Appendix A. As the optical processor could deterministically achieve only two solutions, $|20\rangle$ and $|02\rangle$, the number of bits in the output of the last stage was the same as in the one-bit processor. This is 4 for strategies associated with possibly three different outputs from the previous module. The total number of binary variables in the optimization of the three-module trit processor was 44. We found strategies of each module that achieved the minimum number of errors, which is shown in Table \ref{Tabstage3}. 

For finding the optimal four-module trit processor strategy with Fixstars Amplify, we needed (for license restriction) to limit the number of binary variables to 48. We did it by guessing the strategies of the first and last modules. The solver found the optimal strategies of the intermediate modules as shown in Table \ref{Tabstage4}.

\section{Discussion: extension of the results}\label{sec:lowranks}

In this paper, we have studied theoretically and experimentally the separation between performance of quantum and classical processors of the same structure tackling the problem of approximation of certain functions of input variables when the variables are available only locally to different modules and the communication between the modules is restricted. The quantum processor could potentially use intermediate measurements and effectively collect and pass classical communication. Therefore, the quantum processor could simulate a classical one. However, the opposite is not true. In particular, binary functions that can be achieved by the quantum sequential processors with qubit or qutrit communication can only be approximated by the output of the corresponding classical processors. We showed that existing silicon optics technology allows us to demonstrate this separation.

To prove the separation in the cases tailored to our experiment, we needed to develop a general theoretical framework that allows us to put strict boundaries on the performance of classical processors. We did it in two ways. The first one is based on the properties of specific target binary functions related to the output of our experiment. The second way of finding classical bounds does not rely on particular structure of the target function. It comes from finding the minimimum Hamming distance between a given binary function and a set of binary functions that satisfy certain conditions. 


In our case, the condition is to be the output of a sequential processor with communication between the modules restricted to a link with given dimensionality. The restriction has a clear interpretation when we reshape the function in the form of a matrix. The rows of the matrix correspond to the processor response to given outputs of one module. As the number of symbols in the output is restricted, the number of different rows is also restricted. The number of different rows in a binary matrix corresponds to its rank. So, the problem of finding the optimal output of a two-module processor with limited communication to an arbitrary binary function directly corresponds to the problem of approximation of an arbitrary binary matrix by a binary matrix of specific rank. The optimality is here defined by the Hamming distance. We reduce this optimization to an Ising problem and solve it in certain cases by existing Ising solvers.

Unlike low-rank approximation of real and complex matrices that can be done by reducing rank in the singular value decomposition, the low-rank approximation of binary matrices is known to be NP hard. A side result of our work is the reduction of this problem to the generalized Ising problem. The general Ising problem is also known to be hard in general, but dedicated approximate solvers exist and are still improving. Notably, quantum computers are candidates for interesting heuristics in the resolution of this problem.

Another restriction on approximation of the general binary target function of a number of variables can be the output of a multimodule processor. In this case, we approximate a general tensor (indexed by variables) in terms of a concatenation of tensors of specific ranks. Finding an optimal approximation, in terms of the Hamming distance, is again a difficult task. We can explicitly reduce it to a generalized Ising problem and approach it with existing Ising solvers.

The third related application of the techniques developed in this paper is the binary matrix completion. We could ask, what is a good fixed rank binary matrix completion of a given binary matrix with restricted number of known entries. Again, we tackle this problem by creating the binary matrix with highest possible correlation function with the existing part of the reference matrix and with constraints on the ranks in different reshapings of the matrix. 

Finally, the direct application of our techniques is the logic and sequential logic design. In fact, an arbitrary binary target function of a number of variables can be interpreted as the truth table for some logical formulas of the corresponding Boolean variables. The optimal construction of logical circuits or circuits that meet specific conditions, such as being distributed or sequential, is a highly desirable and unsolved problem. Although approximations exist and are widely used, techniques, such as the one proposed in this paper, can contribute to design improvement. 

\begin{acknowledgments}
We thank Marek \.Zukowski for valuable discussions. This work was supported by JST Moonshot $R\&D$, Grant No. JPMJMS226C and Grant No. JPMJMS2061, JST ASPIRE,
Grant No. JPMJAP2427, JST COI-NEXT Grant No. JPMJPF2221, JST ERATO Grant No. JPMJER2402, and JSPS KAKENHI Grant No. JP24K00559, JST SPRING Grant No. JPMJSP2123.
\end{acknowledgments}

\appendix
\section{Classical bounds on correlations}\label{app:classical}

In this appendix, we establish bounds on the correlations with the target functions given in Eqs.  (\ref{targetbit}) and (\ref{target4}), which are achievable by outputs of sequential classical processors with one-bit and one-trit communication, respectively, as shown in Figs. \ref{fig:resultsa}--\ref{fig:resultsb}. The maximum correlation is uniquely defined by the minimum number of errors between the output of the classical processor and the target function. In the following parts, we establish these results as propositions.

\subsection{Bound for one-bit processor}

\begin{proposition}\label{prop:bit3}
The number of errors between the output of the three-stage processor with one-bit communication and the nonzero part of the target function (\ref{targetbit}) cannot be less than 8.
\end{proposition}

\begin{proof}
In the proof, we explicitly use the properties of the target function (\ref{targetbit}). The rows of this matrix correspond to the processor action depending on the first-stage output. As the first module can have only two different symbols, the target reshuffled as a $4\times 16$ matrix needs to be of rank 2. So, two lines need to be identical. Let these lines be first and third (the same can be repeated for all other choices). We add the two rows, obtaining a binary vector of length 16.

We reshape this vector into a matrix $4\times 4$. It needs to be rank 2, as the rows describe the processor's actions depending on the output from the second stage, which can output at most two symbols.  However, we observe that each row is different and that the minimum number of different bits between each pair of rows (minimum Hamming distance) is 2. We need to correct two bits in two rows. In this way, we necessarily introduce four errors.

Finally, we have the remaining lines 2 and 4 of the target reshaped into the matrix $4\times 16$. When we reshape each of these lines into $4\times 4$ matrices, for the same reason as before, we need to correct four additional bits.

This justifies that the number of errors that we need to introduce to the target function (\ref{targetbit}) to make it compatible with the class of functions output by the discussed classical sequential processor is not less than 8. 

Performing simulated annealing combinatorial optimization on Fixstars Amplify 
we find a strategy for each module of the processor such that the final output  approximates the non-zero part of the target with exactly eight errors. The example is shown in Table \ref{table:bit3}. 

This concludes the proof. The minimum eight errors correspond to the maximum value of the correlation between the output and the target equal to 0.25.
\end{proof}

Actually, because the size of the problem is relatively small, we could guess the solution with only eight errors. However, we perform the simulation on the annealer to demonstrate that our numerical optimization tool works well.

\subsection{Bound for three-module trit processor}

\begin{proposition}\label{prop:trit3}
The number of errors between the output of the three-stage processor with one-trit communication and the nonzero part of the target function (\ref{target4}) (lines 1 to 4) cannot be less than 2.
\end{proposition}

\begin{proof}
The target function is the first four rows of Eq. (\ref{target4}). As the rows of the matrix correspond to the functions in modules 2 and 3 given one of three symbols output from module 1, the rank of the matrix has to be equal to 3. Because the matrix is binary, the rank restriction requires that two of the rows be identical. 

Without loss of generality, assume that the first and third rows are identical (the procedure can be repeated in all other cases). We complete the undefined elements of the first row by the defined elements of the third row. %
We observe that we have four blocks of size 4. Moreover, all are different, and the number of different bits between each pair of the four blocks is equal to 2. %
We can reduce the number of these 4-size blocks to three by changing at least two bits. So, the smallest number of changes we would need is 2. 

We can give an explicit example of a processor that approximate the target function with only two errors, as shown in Table \ref{Tabstage3}.

This completes the proof of the proposition. The two errors correspond to the maximum value of the correlation function of 0.4375.

\end{proof}

\subsection{Bound for four-module trit processor}

\begin{proposition}\label{prop:trit4}
The number of errors between the output of the four-stage processor with one-trit communication and the non-zero part of the target function (\ref{target4}) cannot be less than 16.
\end{proposition}

\begin{proof}
The proof uses the properties of the target function (\ref{target4}) explicitly. 

We reshape the target into the $4\times64$ matrix. The rows correspond to the processor action depending on the first-stage output. As it can have only three different symbols, the target reshuffled as a $4\times64$ matrix must be of rank 3. So, two lines need to be identical. Let these lines be the first and second (the same can be repeated for all other choices). We add the two rows, obtaining a binary function of length 64.

We reshape this vector into a matrix $4\times 16$. It needs to be rank 3, as the rows describe the processor's actions depending on the output from the second stage, which can contain at most three symbols. So, two rows of this matrix have to be the same. However, we observe that the minimum number of different bits between each pair of rows (minimum Hamming distance) is 8. Correcting some entries in one of the rows, we necessarily introduce eight errors in the original target function.

We observe that the remaining lines of the $4\times16$ matrix contain four different blocks of length 4. To reduce this number to three different blocks of length 4, we need to introduce two additional errors per line. This means that four additional errors are required to correct these parts.

Finally, we have the remaining lines 3 and 4 of the target reshaped into the $4\times 64$ matrix. When we reshape each of these lines into $4\times16$ matrices, for the same reason as before, we expect that each of them is of rank 3. So, the two lines need to be identical. We observe that when we add two complementary lines we get vectors of length 16 consisting of four blocks of size 4. To reduce this number to 3, we need to introduce two errors per each of the two rows.

This justifies that the number of errors that we need to introduce to the target function to make it compatible with the class of functions output by the discussed classical sequential processor is not smaller than 16. 

Performing simulated annealing combinatorial optimization on Fixstars Amplify,  
we find an optimal strategy for each module of the processor such that the final output approximates the non-zero part of the target with exactly 16 errors. The example is shown in Table \ref{Tabstage4}. 

This concludes the proof. The minimum number of errors 16 correspond to the maximum value of the correlation function equal to 0.375.
\end{proof}

\section{Existence of classical processor with a given output}\label{app:existence}

For simplicity, assume that the number of processing modules is $N=3$. We can solve Task 1. by checking the number of different blocks of size 4, and 16 in the target function. There are 16 blocks of size 4. They are indexed by sequences of bits $(X_1X_2X_3X_4)$ and defined as 
$$
t^{[4]}(X_1\ldots X_4,X_5X_6) = T(X_1,\ldots,X_6).
$$
There are four blocks of size 16. They are indexed by sequences of bits $(X_1X_2)$ defined as 
$$
t^{[16]}(X_1X_2,X_3\ldots X_6) = T(X_1,\ldots,X_6)
$$\\

\begin{theorem}
\label{theo:diff_rows}
    The sequential processor with three stages and one-bit communication can output functions $O(X_1,...,X_6)$ with at most two different rows in the matrix $O^{[4]}(X_1X_2,X_3...X_6)$ and at most two different rows in the reshuffled matrix $O^{[16]}(X_1...X_4,X_5X_6)$.
\end{theorem}

\begin{proof} 
There are four possible outputs (one output for each $X_1X_2$) from the first processor but only two binary symbols $0$ or $1$. For each symbol, there is one and only one succession [the row of the matrix $O^{[4]}(X_1X_2,X_3...X_6)$]. This is because the following stage cannot distinguish the same symbols output for different inputs of the first stage. So, there can be only two different rows in matrix $O^{[4]}$. 
Similarly, there are 16 outputs (one output for each sequence $(X_1...X_4)$) from the second processor, but only two binary symbols. For each symbol, there is a single succession. 
The row of the matrix $O^{[16]}(X_1...X_4,X_5X_6)$ for a chosen value of $(X_1...X_4)$ is equal to one of these successions. We recognize that the target function consists of only two distinct rows of $O^{[16]}$ and only two distinct rows of $O^{[4]}$. Thus, we can associate symbols $0$ or $1$ as the output of the first processor in the positions corresponding to different distinct rows. Similarly, we associate $0$ or $1$ as the outputs of the second processor in the positions corresponding to distinct rows of $O^{[16]}$. 
\end{proof}

The argument can be easily extended to processors with more modules and different limits on communication between them.

\end{document}